\PassOptionsToPackage{x11names}{xcolor}
\documentclass[11pt]{article}
\usepackage[margin=1in]{geometry}

\usepackage{amsthm,amsmath,amssymb}
\usepackage{todonotes}
\usepackage{setspace}
\usepackage[T1]{fontenc}
\usepackage{hyperref}
\usepackage{cleveref}
\usepackage{geometry,textpos}
\usepackage{mathpazo}
\usepackage{euler}

\hypersetup{
    hidelinks,
    colorlinks,
    linkcolor  = [rgb]{0.0, 0.2, 0.6},   
    citecolor  = [rgb]{0.0, 0.5, 0.5},   
    urlcolor   = [rgb]{0.5, 0.0, 0.0}    
}

\usepackage[backend=biber, style=alphabetic,maxbibnames=99,backref=true, isbn=false]{biblatex}
\bibliography{refs}
\newbibmacro{string+doiurl}[1]{%
  \iffieldundef{doi}{%
    \iffieldundef{url}{#1}{%
        \href{\thefield{url}}{#1}
    }
  }{
    \href{https://dx.doi.org/\thefield{doi}}{#1}
    }
}
\DeclareFieldFormat*{title}{\usebibmacro{string+doiurl}{\mkbibemph{#1}}}

\usepackage{tikz}
\usetikzlibrary{decorations.pathreplacing,calligraphy}
\usetikzlibrary{positioning}
\usetikzlibrary{decorations}
\usetikzlibrary{decorations.pathmorphing}
\usetikzlibrary{shapes, fit, arrows.meta}
\usetikzlibrary{calc}

\newcommand\Z{\mathbb{Z}}

\newcommand\Q{\mathbb{Q}}

\newtheorem{theorem}{Theorem}[section]

\newtheorem{lemma}[theorem]{Lemma}
\newtheorem{claim}[theorem]{Claim}

\theoremstyle{definition}
\newtheorem{definition}[theorem]{Definition}
\newtheorem{remark}[theorem]{Remark}

\numberwithin{equation}{section}

\newcommand{\vc}{\mathrm{vc}}
\newcommand{\td}{\mathrm{td}}
\newcommand{\ptw}{\mathrm{ptw}}
\newcommand{\tw}{\mathrm{tw}}
\newcommand{\pw}{\mathrm{pw}}
\newcommand{\ppw}{\mathrm{ppw}}
\newcommand{\poly}{\mathrm{poly}}
\newcommand{\CI}{\mathsf{ColSub}}
\newcommand{\Hom}{\mathsf{Hom}}
\newcommand{\Sub}{\mathsf{\#Sub}}
\newcommand{\rep}{\mathrm{rep}}
\newcommand{\T}{\mathcal{T}}

\title{Monotone Bounded-Depth Complexity of\\Homomorphism Polynomials}

\author{
    {C.S. Bhargav \thanks{Indian Institute of Technology Kanpur, India. Email: \texttt{bhargav@cse.iitk.ac.in}}}\and
    {Shiteng Chen \thanks{Key Laboratory of System Software (Chinese Academy of Sciences) and State Key Laboratory of Computer Science, Institute of Software, Chinese Academy of Sciences and University of Chinese Academy of Sciences, Beijing, China. Email: \texttt{chenst@ios.ac.cn}. Supported by National Key R \& D Program of China (2023YFA1009500), NSFC 61932002 and NSFC 62272448}}\and
    {Radu Curticapean \thanks{University of Regensburg, Germany and IT University of Copenhagen, Denmark. Email: \texttt{radu.curticapean@ur.de}. Funded by the European Union (ERC, CountHom, 101077083).}}\and
    {Prateek Dwivedi\thanks{IT University of Copenhagen, Denmark. Email: \texttt{prdw@itu.dk}. Funded by the \emph{Independent Research Fund Denmark} (FLows 10.46540/3103-00116B). Also supported by Basic Algorithms Research Copenhagen (BARC), Villum Investigator Grant 54451.}}
}

\date{}

\begin{document}

\maketitle

\begin{textblock}{5}(10.3, 10.2) \includegraphics[width=60px]{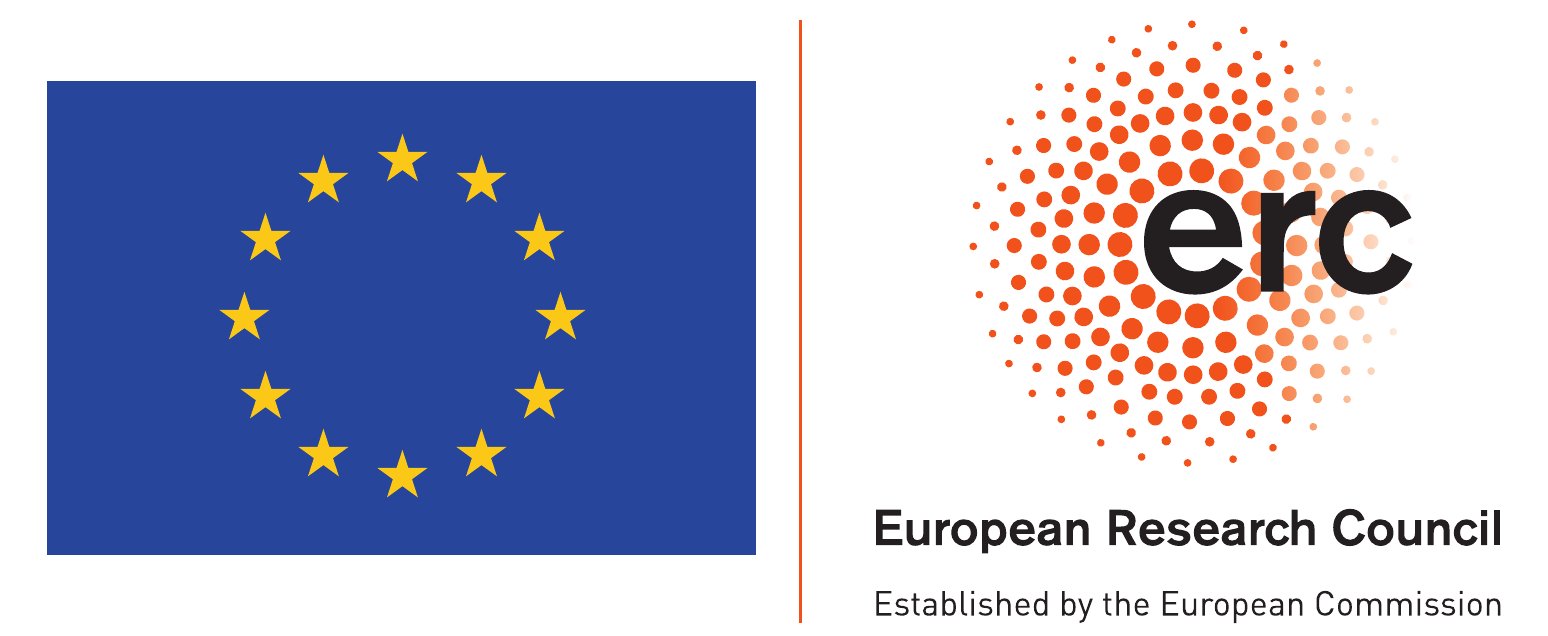} \end{textblock}

\begin{abstract}
    For every fixed graph $H$, it is known that homomorphism counts from $H$ and colorful $H$-subgraph counts can be determined in $O(n^{t+1})$ time on $n$-vertex input graphs $G$, where $t$ is the treewidth of $H$.
    On the other hand, a running time of $n^{o(t / \log t)}$ would refute the exponential-time hypothesis.
    Komarath, Pandey and Rahul (Algorithmica, 2023) studied algebraic variants of these counting problems, i.e., homomorphism and subgraph \emph{polynomials} for fixed graphs $H$.
    These polynomials are weighted sums over the objects counted above, where each object is weighted by the product of variables corresponding to edges contained in the object.
    As shown by Komarath \emph{et al.}, the \emph{monotone} circuit complexity of the homomorphism polynomial for $H$ is $\Theta(n^{\tw(H)+1})$.
    
    In this paper, we characterize the power of monotone \emph{bounded-depth} circuits for homomorphism and colorful subgraph polynomials. 
    This leads us to discover a natural hierarchy of graph parameters $\tw_\Delta(H)$, for fixed $\Delta \in \mathbb N$, which capture the width of tree-decompositions for $H$ when the underlying tree is required to have depth at most $\Delta$. We prove that monotone circuits of product-depth $\Delta$ computing the homomorphism polynomial for $H$ require size $\Theta(n^{\tw_\Delta(H^{\dagger})+1})$,
    where $H^{\dagger}$ is the graph obtained from $H$ by removing all degree-$1$ vertices. This allows us to derive an optimal depth hierarchy theorem for monotone bounded-depth circuits through graph-theoretic arguments.
    
\end{abstract}

\section{Introduction}
Counting and deciding the existence of patterns in graphs plays an important role in computer science. 
In theoretical computer science, pattern counting was among the first problems to be investigated in Valiant's seminal paper on the class $\mathsf{\#P}$ of counting problems~\cite{Val1979a}, which showed $\mathsf{\#P}$-hardness for the permanent of zero-one matrices, a problem that can equivalently be viewed as counting perfect matchings in bipartite graphs. 

\paragraph*{Counting small patterns.}

In many applications, the pattern is smaller in comparison to the target graph.
Curticapean and Marx~\cite{CM2014} modeled this setting by classifying the complexity of counting subgraphs from fixed pattern classes: Given any fixed class of graphs $\mathcal H$, they defined a problem $\Sub(\mathcal H)$ that asks, given a graph $H\in \mathcal H$ and a general graph $G$, to count the $H$-subgraphs in $G$.
The parameter is $|V(H)|$.
This problem is known to be polynomial-time solvable when the graphs in $\mathcal H$ do not contain arbitrarily large matchings---if they do contain arbitrarily large matchings, then $\Sub(\mathcal H)$ with parameter $|V(H)|$ is complete for the parameterized complexity class $\mathsf{\#W[1]}$, i.e., the analogue of $\mathsf{\#P}$ in parameterized complexity.

An analogously defined problem $\#\Hom(\mathcal H)$ of counting \emph{homomorphisms} with patterns drawn from $\mathcal H$ was also classified by Dalmau and Jonsson~\cite{DJ2004}. Here, the tractability criterion is a constant bound on the \emph{treewidth} of graphs in $\mathcal H$, a measure of the ``tree-likeness'' of $H$: The problem $\#\Hom(\mathcal H)$ is polynomial-time solvable when all graphs in $\mathcal H$ admit a constant upper bound on their treewidth, and the problem is $\mathsf{\#W[1]}$-hard otherwise with respect to the parameter $|V(H)|$. Here, a homomorphism from $H$ to $G$ is a function $h: V(H)\to V(G)$ such that $uv\in E(H)$ implies $h(u)h(v)\in E(G)$. Homomorphism counts from small patterns find direct applications in database theory, where they capture answer counts to so-called conjunctive queries~\cite{CM2016}. It was also shown that $\#\Hom(\mathcal H)$ captures the complexity of other pattern counting problems, including that of counting subgraphs~\cite{CDM2017}: In a nutshell, (i) many pattern counting problems can be expressed as unique linear combinations of homomorphism counts from graphs $H$, and (ii) in many models of computation, such linear combinations turn out to be precisely as hard as their hardest terms. This motivates understanding the complexity of these individual terms, i.e., of homomorphism counts.

\paragraph{Lower bounds under ETH.}
Following the classification of $\Sub(\mathcal H)$ and $\#\Hom(\mathcal H)$ under parameterized complexity assumptions, almost-tight quantitative bounds were obtained under the exponential-time hypothesis~\cite{IP2001, LMS2011}:
For any graph $H$, there is an $O(n^{\tw(H)+1})$ time algorithm for counting homomorphisms from $H$ into $n$-vertex target graphs, and assuming the exponential-time hypothesis, Marx ruled out $n^{o(\tw(H) / \log \tw(H))}$ time algorithms~\cite{Mar2010}, recently revisited in~\cite{KMPS2024,CDNW2025}.
Through connections between homomorphism counts and other pattern counts~\cite{CDM2017}, these bounds translate directly to other counting problems.
For example, there is an $O(n^{\vc(H)})$ time algorithm for counting $H$-subgraphs in an $n$-vertex graph $G$, where $\vc(H)$ is the vertex-cover number of $H$.
As a consequence of the lower bound on counting homomorphisms, the exponential-time hypothesis rules out $n^{o(\vc(H) / \log \vc(H))}$ time algorithms~\cite{CDM2017}.

Thus, some slack remains between the known upper and conditional lower bounds on the exponents of pattern counting problems, even asymptotically:  
It would be desirable to settle the $\log$-factor in the exponent of the running time. Moreover, one might also dare to ask for the precise exponent for concrete finite graphs $H$, such as $K_3$ (which amounts to triangle counting) or $K_4$ or $C_6$. Finally, let us stress that the lower bounds on the exponent are conditioned on the exponential-time hypothesis, an assumption that is \emph{a priori} stronger than $\mathsf P \neq \mathsf{NP}$.

\paragraph*{From counting problems to polynomials.}

Valiant's seminal papers~\cite{Val1979, Val1980} studied the problem of counting perfect matchings from the perspective of both counting and algebraic complexity. In our paper, following the work of Komarath, Pandey, and Rahul~\cite{KPR2023}, we consider an algebraic version of pattern counting problems.

For undirected graphs $H$ and $n\in\mathbb{N}$, we consider
the homomorphism polynomial $\Hom_{H,n}$ on variables $x_{i,j}$ for $i,j\in[n]$
and its set-multilinear version, the colorful subgraph polynomial $\CI_{H,n}$ on variables
$x_{i,j}^{(e)}$ for $i,j\in[n]$ and $e\in E(H)$. The latter can often be handled more easily in proofs, while complexity results can be transferred between these two polynomials.

\begin{remark}
We deviate slightly from the notation used by Komarath, Pandey, and Rahul~\cite{KPR2023}, who defined a polynomial $\mathsf{ColIso}_{\mathcal{H},n}$ with variable indices that differ from ours. Our polynomial $\CI_{H,n}$ and their polynomial $\mathsf{ColIso}_{\mathcal{H},n}$ can be obtained from each other by renaming variables. We consider our notation more intuitive, as it can be obtained from the homomorphism polynomial more directly and also highlights the set-multilinearity of the polynomial.    
\end{remark}

The polynomial $\Hom_{H,n}$ can be viewed as the weighted homomorphism count from $H$ into a complete $n$-vertex graph with generic indeterminates as edge-weights. 
Similarly, $\CI_{H,n}$ can be viewed as counting the color-preserving homomorphism count from a colorful graph $H$ into a complete graph with indeterminate edge-weights and $n$ vertices per color class.
Formally, these multivariate polynomials are defined as
\begin{align*}
\Hom_{H,n} & =\sum_{f:V(H)\to[n]}\prod_{uv\in E(H)}x_{f(u),f(v)},\\
\CI_{H,n} & =\sum_{f:V(H)\to[n]}\prod_{uv\in E(H)}x_{f(u),f(v)}^{(uv)}.
\end{align*}

The complexity of multivariate polynomials is commonly studied via algebraic circuits, first formalized by Valiant \cite{Val1979}. The efficiency of a circuit is usually quantified by its \emph{size} (number of edges and gates) and its \emph{depth} (number of layers). The size captures the total number of operations the circuit performs, and the \emph{depth} roughly corresponds to \emph{parallelism}. For convenience, we will instead consider the \emph{product-depth}, which is the number of `multiplication layers' in the circuit. Refer to \Cref{sec:prelims} for formal definitions. These efficiency measures define natural algebraic circuit complexity classes, for instance, $\mathsf{VP}$ --- polynomials of polynomially bounded degree computable by circuits of size $\poly(n)$, and $\mathsf{VNP}$ --- polynomials which can be expressed as a hypercube sum of $\mathsf{VP}$ circuits.

Valiant established that the permanent polynomial is complete for the class $\mathsf{VNP}$, while its sibling, the determinant polynomial, is complete for a seemingly smaller class, $\mathsf{VBP}$—the class of polynomially bounded algebraic branching programs (see \Cref{def:abp}). In a recent line of work \cite{CLV2021, DMM+2016, MS2018}, homomorphism polynomials have been used to obtain natural polynomials which are complete for several well-studied algebraic circuit classes.

\paragraph{Monotone circuits.}
\emph{Monotone} circuits over a field like $\Z$ are circuits that do not use negative constants, and hence computations performed by them cannot feature cancellations (\Cref{def:monotone-ckt}). Several important techniques for proving upper bounds on the complexity of polynomials (e.g., dynamic programming) directly yield monotone circuits. Compared to general computational models, lower bounds for monotone computation are much better understood, and many exponential lower bounds~\cite{Sch1976, Val1980, JS1982, RY2011, GS2012, CKR2022} and strong algebraic complexity class separations~\cite{Sni1980, HY2016, Yeh2019, Sri2020} are known. As a striking example, monotone variants of the algebraic complexity classes $\mathsf{VP}$ and $\mathsf{VNP}$ are proven to be different \cite{Sri2020, Yeh2019}. In contrast, Hrubes~\cite{Hru2020} showed that strong-enough monotone lower bounds of a special kind (called $\epsilon$-\emph{sensitive}) imply unconditional lower bounds for general arithmetic circuits! 

In a fascinating recent work, Komarath, Pandey and Rahul~\cite{KPR2023} studied the monotone arithmetic circuit complexity of the polynomials $\Hom_{H,n}$ and discovered that this complexity is completely determined by the treewidth of the pattern graph $H$. More precisely, they show that the smallest monotone circuit computing $\Hom_{H,n}$ is of size $\Theta(n^{\tw(H)+1})$. Similarly, they show that algebraic branching programs for $\Hom_{H,n}$ are of size $\Theta(n^{\pw(H)+1})$, where $\pw(H)$ is the \emph{pathwidth} of $H$, a linear version of treewidth. Moreover, they also consider the monotone formula complexity of $\Hom_{H,n}$ and show that it is  $\Theta(n^{\td(H)+1})$, where $\td(H)$ is the \emph{treedepth} of $H$, the minimum height of a tree on vertex set $V(H)$ that contains all edges of $H$ in its tree-order.
These results together show that, when considering homomorphism polynomials for fixed patterns $H$, the power of monotone computation is precisely characterized by graph-theoretic quantities of $H$: For natural and well-studied monotone computational models, the precise exponent $c_H$ in the complexity $\Theta(n^{c_H})$ is the value of a natural and well-studied graph parameter of $H$. 

\paragraph*{Our results: Monotone bounded-depth models.}
In this paper, we investigate whether the correspondence between monotone circuit complexity and graph parameters can also be established for another restriction of monotone circuits, namely \emph{bounded-depth} monotone circuits: Are there natural graph parameters that dictate the bounded-depth monotone complexity of $\Hom_{H,n}$?

Bounded-depth circuits are of central importance in algebraic complexity due to the phenomenon of \emph{depth reduction}. In a sequence of works~\cite{VSBR1983,AV2008,Koi2012,Tav2015}, it was shown that any algebraic circuit of size $s$ computing a polynomial of degree $d$ can also be simulated by a product-depth $\Delta$ circuit of size $s^{O(d^{1/\Delta})}$. If the circuit was monotone to begin with, the resulting bounded-depth circuit is also monotone. 
Depth reduction also implies that strong enough lower bounds for bounded-depth circuits will lead to general circuit lower bounds -- an exceptionally hard open question. A lot of work has gone into proving strong bounded-depth circuit lower bounds (see~\cite{Sap2015} for a survey). Recently, following the breakthrough result of Limaye, Srinivasan and Tavenas~\cite{LST2021} superpolynomial lower bounds have been shown for bounded-depth circuits (also see~\cite{BDS2024,AGK+2023}).

Studying the monotone bounded-depth complexity of $\Hom_{H,n}$ naturally leads us to define bounded-depth versions of treewidth, the \emph{$\Delta$-treewidth} $\tw_\Delta(H)$ for any fixed $\Delta \in \mathbb N$. These graph parameters ask to minimize the maximum bag size over all tree-decompositions of $H$, however with the twist that only tree-decompositions with an underlying tree of height at most $\Delta$ are admissible. Their values interpolate between $|V(H)|-1$ (when only height $1$ is allowed) and $\tw(H)$ (when no height restrictions are imposed), and they are connected to the vertex-cover number in the special case $\Delta = 2$ (see \Cref{sec:bound-dp-tw}). Bounded-depth variants of treewidth implicitly appear in balancing techniques for tree-decompositions~\cite{CIP2016,BH1998}, and the $\Delta$-treewidth of paths also appears implicitly in divide-and-conquer schemes for iterated matrix multiplication in a bounded-depth setting~\cite{LST2021}. A recent work of Adler and Fluck~\cite{AF2024} studied a notion that bounds the width and depth simultaneously, which they call bounded depth treewidth. Our notion of $\tw_\Delta(H)$ only bounds the height of the tree decomposition to $\Delta$. In particular, for a fixed $\Delta$, there is always a tree decomposition of height $\Delta$ for any graph $H$, but with a possibly large treewidth.

We show that the $\Delta$-treewidth of graphs completely characterizes the complexity of $\Hom_{H,n}$ for monotone circuits of product-depth at most $\Delta$. For technical reasons described later, it is however not the $\Delta$-treewidth of $H$ itself that governs the complexity, but rather the $\Delta$-treewidth of the graph $H^{\dagger}$ obtained by removing all vertices of degree at most $1$. We call this the \emph{pruned} $\Delta$-treewidth $\ptw_{\Delta}$ of a graph $H$. 
Our main result is as follows:

\begin{theorem}
\label{thm: main}
    Let $H$ be a fixed graph and let $\Delta$ and $n$ be natural numbers. 
    Then the polynomials $\Hom_{H,n}$ and $\CI_{H,n}$ have monotone circuits of size $O(n^{\ptw_\Delta(H)+1})$ and product-depth $\Delta$.
    Moreover, any monotone circuit of product-depth $\Delta$ has size $\Omega(n^{\ptw_\Delta(H)+1})$.
\end{theorem}

Note that any fixed pattern graph $H$ on $k$ vertices gives a homomorphism polynomial on $n^{k}$ monomials, which has a trivial $\poly(n)$ sized monotone circuit of depth two. We stress that we wish to determine the precise \emph{exponent} in this polynomial size: In general, $k$ could be much larger than the pruned $\Delta$-treewidth of $H$.

We also study the related computational model of algebraic branching programs (ABPs). An ABP is a directed acyclic graph with a source vertex $s$ and a sink $t$, with edges between vertices labeled by variables or constants. The \emph{size} of the ABP is the total number of vertices in the graph, and the length of the longest path from $s$ to $t$ is its length. We refer the reader to \Cref{sec:prelims} for a formal definition of an ABP and its computation.

For a graph $H$, we define a new graph parameter called $\Delta$-\emph{pathwidth} $\pw_{\Delta}(H)$, which asks to minimize the bag size over all \emph{path decompositions} of $H$ where the underlying path is of length $\Delta$. We defer the formal definition these graph parameters to \Cref{sec:bound-dp-tw}. Similar to the case of circuits, we show that the $\Delta$-pathwidth of the pruned graph $H^{\dagger}$ (obtained by removing degree $1$ vertices from $H$), which we call the \emph{pruned} $\Delta$-pathwidth $\ppw_{\Delta}$ of $H$ characterizes the monotone ABP of length $\Delta$ computing $\Hom_{H,n}$.

\begin{theorem}
\label{thm: abp}
    Let $H$ be a fixed graph and let $\Delta$ and $n$ be natural numbers such that $\Delta \geq |E(H)|$. Then the polynomials $\Hom_{H,n}$ and $\CI_{H,n}$ can be computed by monotone algebraic branching programs of size $O(n^{\ppw_\Delta(H)+1})$ and length $\Delta$. Moreover, any monotone algebraic branching program of length $\Delta$ has size $\Omega(n^{\ppw_\Delta(H)+1})$.
\end{theorem}

For a length-$\Delta$ ABP to compute the polynomial $\Hom_{H,n}$ (or $\CI_{H,n}$), its length needs to be at least the degree of the polynomial (which is $|E(H)|$). Otherwise, we cannot even compute a single monomial. We note that the above theorem also implies a bound on the ABP \emph{width}.

For a fixed pattern graph, it was shown in \cite{KPR2023} that $\Hom_{H,n}$ and $\CI_{H,n}$ have the same ``monotone complexity''. We observe that the reduction holds even in the bounded-depth (bounded-length) case (\Cref{lem: col-hom-reduction}), so we prove our results only for $\CI_{H,n}$. The upper bounds of \Cref{thm: main} and \Cref{thm: abp} are shown in \Cref{sec:upper-bound} and \Cref{sec:lower-bound}, respectively.

\paragraph*{Our results: Monotone depth hierarchy.}

Finally, by turning our attention to pattern graphs of non-constant size, we can prove a depth hierarchy theorem for monotone circuits: Using the tight characterization from the above theorem and the properties of pruned $\Delta$-treewidth, we are able to obtain the following.

\begin{theorem}
\label{thm:dh}
    For any natural numbers $n$ and $\Delta$, there exists a pattern graph $H_{\Delta}$ of size $\Theta(n)$ such that $\CI_{H_{\Delta},n}$ can be computed by a monotone circuit of size $\poly(n)$ and product-depth $\Delta+1$, but every monotone circuit of product-depth $\Delta$ computing the polynomial needs size $n^{\Omega(n^{1/\Delta})}$. 
\end{theorem}

    By general depth-reduction results, any monotone circuit of size $\poly(n)$ computing a polynomial of degree $n$ can be flattened to a monotone circuit of product-depth $\Delta$ and size $n^{O(n^{1/\Delta})}$, showing that the above theorem is optimal.
    
    We also note that a similar \emph{near-optimal} statement with a lower bound of $\exp(n^{\Omega(1/\Delta)})$ can be obtained from earlier results provided the product-depth $\Delta = o(\log n/\log \log n)$ is small (see~\cite{CELS2018}). Our results improve upon this in two ways: Firstly, our $\Omega(\cdot )$ appears in the first rather than second exponent, thus yielding a stronger and optimal lower bound. Secondly, our results hold for any product-depth $\Delta$.

\section{Preliminaries} \label{sec:prelims}

For a natural number $n \in \mathbb{N}$, we will use $[n]$ to refer to the set $\{1,\ldots,n\}$. We begin with some algebraic complexity and graph-theoretic preliminaries. Readers comfortable with these notions can safely skip this section.

\subsection{Algebraic complexity theory}

Algebraic circuits are analogous to Boolean circuits, where logical operators are replaced with underlying field operations. In this paper, we fix our field to be rationals $\Q$.
\footnote{Our results hold for any field by making appropriate changes to the definition of monotone computation so that cancellations are avoided.} To study the complexity of polynomials, Valiant formulated the algebraic complexity theory \cite{Val1979}. Here, we will only give relevant definitions, but we encourage interested readers to refer to surveys for a more comprehensive overview of the area \cite{Mah2014,Sap2015,SY2009}. 

\begin{definition}[Algebraic Circuits] \label{def:circuits}    
    An \emph{algebraic circuit} $C$ is a \emph{layered} directed acyclic graph with a unique root node, called the output gate. Each leaf node of $C$ is called an input gate and is labeled by one of the variables $x_1, \dots, x_n$ or a field constant.
    Gates are either labeled $+$ or $\times$, and on every path from the output gate to some input gate, these gate types alternate. The \emph{(product) depth} of $C$ is the number of (multiplication) layers in the circuit, while its \emph{size} is the number of vertices of the underlying graph. The circuit naturally computes a polynomial: a $+ (\times)$ gate computes the sum (product) of the polynomials computed by its children.
\end{definition}

\begin{definition}[Skew Circuits]
\label{def:skew-ckt}
    An algebraic circuit is called \emph{skew} if, for every multiplication gate, \emph{at most one} of its children is an internal (non-input) gate.
\end{definition}

\begin{definition}[Algebraic Branching Programs] \label{def:abp}
    An \emph{Algebraic Branching Program} (ABP) is a directed acyclic graph with edges labeled by a variable or a field constant. It has a designated \emph{source} node $s$ (of in-degree $0$) and a \emph{sink} node $t$ (of out-degree $0$). A path from $s$ to $t$ computes the product of all edge labels along the path. The polynomial computed by the ABP is the sum of the terms computed along all the paths from $s$ to $t$. If all the constants in the ABP are non-negative, the ABP is \emph{monotone}. The \emph{size} of an ABP is the total number of vertices in the graph, and the length of the longest path from $s$ to $t$ is the \emph{length} of the ABP. 
\end{definition}

Monotone computation, which is free of cancellations, can be simulated by algebraic circuits (branching programs) by restricting the choice of field constants. 

\begin{definition}[Monotone Circuits and ABPs] \label{def:monotone-ckt}
    A \emph{monotone} circuit (ABP) is an algebraic circuit (ABP) where all the field constants are non-negative. The circuit (ABP) computes a \emph{monotone polynomial}, where coefficients of all the monomials are non-negative.
\end{definition}

\begin{remark}
\label{rem:abp-skew}
    It is not hard to show that skew circuits and ABPs are essentially the same model, up to constant factors (e.g., see the discussion in \cite{Mah2014}). In particular, an ABP of size $s$ and length $\Delta$ can be converted to a skew circuit of size $O(s)$ and product-depth $\Delta$. If the original ABP is monotone, the skew circuit is monotone as well.
\end{remark}

Finally, we define \emph{parse trees} to analyze the computation of each monomial in a circuit. The notion has appeared in several previous results~\cite{AJMV1998, JS1982, KPR2023, MP2008, Ven1992}. 

\begin{definition}[Parse Trees] \label{def:parse-trees}
    A \emph{parse tree} $\T$ of an algebraic circuit $C$ is obtained as follows:
    \begin{itemize}
        \item We include the root gate of $C$ in $\T$.
        \item For every $+$ gate in $\T$, we arbitrarily include \emph{any one} of its children in $\T$.
        \item For every $\times$ gate in $\T$, we include \emph{all} of its children in $\T$.
    \end{itemize}
    We call a parse tree \emph{reduced} if we ignore every $+$ gate, and its parent and (only) child are directly connected by an edge.
\end{definition} 

\begin{remark}
    It is easy to see that every (reduced) parse tree is associated with a monomial of the polynomial computed by the circuit. For a (reduced) parse tree $\T$, let $\mathrm{val}(\T)$ be its output. Then the polynomial computed by the circuit $C$ is $\sum_{\T} \mathrm{val}(\T)$, where the sum is over all the parse trees of $C$. From here on, whenever we use the term parse tree, we mean the reduced parse tree.
\end{remark}

\subsection{Graph theory}

In the following, let $H$ be a graph. A vertex cover of $H$ is a subset $C \subseteq V(H)$ such that for every edge $e\in E(H)$, some vertex $v\in C$ is an endpoint of $e$.
The vertex-cover number of $H$ is the minimum size of a vertex-cover in $H$.

\begin{definition}[Tree-decomposition and treewidth] \label{def:tree-decomp}
    A {tree-decomposition} of $H$ is a tree $T$ whose vertices are annotated with \emph{bags} $\{X_t\}_{t \in V(T)}$, subject to the following conditions:
    \begin{enumerate}
        \item Every vertex $v$ of $H$ is in at least one bag.
        \item For every edge $(u,v)$ in $H$ there is a bag in $T$ that contains both $u$ and $v$.
        \item For any vertex $v$ of $H$, the subgraph of $T$ induced by the bags containing $v$ is a subtree.
    \end{enumerate}
    The \emph{width} of a tree decomposition $T$ is $\max_{t \in V(T)} |X_t|-1$. The \emph{treewidth} of the graph $H$, denoted $\tw(H)$, is the minimum width over all tree decompositions of $H$. 
\end{definition}

\begin{definition}[Path-decomposition and pathwidth] \label{def:path-decomp}
    A {path-decomposition} of a graph $H$ is a tree-decomposition where the underlying tree is a path. The \emph{width} of a path decomposition $P$ is $\max_{t \in V(P)} |X_t|-1$. The \emph{pathwidth} of the graph $H$, denoted $\pw(H)$, is the minimum width over all path decompositions of $H$. 
\end{definition}

We usually consider trees with a designated root vertex.
The height of such a tree is the number of vertices on a longest root-to-leaf path.
We assume trees in tree-decompositions to be rooted in a way that minimizes their height.

\section{Bounded versions of treewidth and pathwidth}
\label{sec:bound-dp-tw}

In this section, we describe $\Delta$-treewidth, our bounded-depth version of treewidth, and $\Delta$-pathwidth, a bounded-length version of pathwidth. We connect these to other known graph parameters. Moreover, we show that full $d$-ary trees of depth $\Delta$, while having $\Delta$-treewidth $1$, have $(\Delta-1)$-treewidth $d-1$. This behavior around the threshold $\Delta$ will ultimately allow us to conclude \Cref{thm:dh}.

\subsection{Connections to other graph parameters}
First, recall the definition of $\Delta$-treewidth from the introduction.
We stress that a tree with a single node has height $1$ according to our definition.

\begin{definition}
    For fixed $\Delta \in \mathbb N$, the $\Delta$-treewidth of a graph $H$, denoted by $\tw_\Delta(H)$, is the minimum width over all tree decompositions of $H$ with underlying tree $T$ of height at most $\Delta$.
\end{definition}

While depth-restricted tree-decompositions did arise before in the literature~\cite{CIP2016}, their depth was not fixed to concrete \emph{constants} in these contexts, but rather to, say, $O(\log |V(H)|)$. In particular, differences between $\Delta$-treewidth and $(\Delta-1)$-treewidth were not considered.

Let us connect the $\Delta$-treewidth of graphs to other graph parameters:
\begin{itemize}
    \item The $1$-treewidth of a graph $H$ is merely its number of vertices $|V(H)|$, as the requirement on the height forces the tree-decomposition to consist of a single bag. On the other extreme, the $|V(H)|$-treewidth of $H$ equals the treewidth of $H$.
    \item The $2$-treewidth is already more curious: For any vertex-cover $C$, a tree-decomposition of height $2$ for $H$ can be obtained by placing $C$ into a root bag $X_r$ that is connected to bags $X_t$ for $t\in V(H)\setminus C$, where $X_t$ contains $t$ and its neighbors, all of which are in $C$. This shows that the $2$-treewidth of $H$ is at most the vertex-cover number $\vc(H)$ of $H$. In fact, the $2$-treewidth of $H$  equals the so-called \emph{vertex-integrity} of $H$ (minus $1$). This graph parameter is defined as $\min_{S\subseteq V(H)}(|S|+\max_{C}|V(C)|)$, where $C$ ranges over all connected components in the graph $H-S$, see~\cite{BES1987,GHK+2025}.
    \item By balancing tree-decompositions~\cite{CIP2016}, a universal constant $c$ can be identified such that, for all graphs $H$ on $k$ vertices, the $c \log k$-treewidth of $H$ is bounded by $4 \tw(H)+3$. That is, at the cost of increasing width by a constant multiplicative factor, tree-decompositions can be assumed to be of logarithmic height.
\end{itemize}

Our upper bound proof will show that vertices of degree $1$ can be removed safely from $H$ without changing the bounded-depth complexity of $\Hom_{H,n}$. This holds essentially because such vertices and their incident edges can be assumed to be present in the leaves of a tree-decomposition; these leaves then do not contribute to the product-depth of the constructed circuit. 
This naturally leads to the notion of \emph{pruned} $\Delta$-treewidth.

\begin{definition}
    The \emph{pruned} $\Delta$-treewidth of a graph $H$, denoted by $\ptw_\Delta(H)$, is the $\Delta$-treewidth of the graph $H$ with all vertices of degree at most $1$ removed.
\end{definition}

We also define analogous \emph{bounded-length} versions of pathwidth.

\begin{definition}
    For fixed $\Delta \in \mathbb N$, the $\Delta$-pathwidth of a graph $H$, denoted by $\pw_\Delta(H)$, is the minimum width over all path decompositions of $H$ with underlying path $P$ of length at most $\Delta$.
\end{definition}

\begin{definition}
    The \emph{pruned} $\Delta$-pathwidth of a graph $H$, denoted by $\ppw_\Delta(H)$, is the $\Delta$-pathwidth of the graph $H$ with all vertices of degree at most $1$ removed.
\end{definition}

\subsection{Full d-ary trees}

We conclude this section by exhibiting a pattern $H$ whose $\Delta$-treewidth shows a strong phase transition that we can exploit in our depth hierarchy theorem: Its $\Delta$-treewidth is low, but even its $(\Delta-1)$-treewidth is high. As it turns out, $H$ can be chosen to be the full $d$-ary tree.

\begin{theorem}
\label{thm:tw-lb}
    Let $\Delta,d$ be positive integers and let $T_{\Delta}$ be the full $d$-ary tree of height $\Delta$. Then $\tw_\Delta(T_{\Delta})=1$ whereas $\tw_{\Delta-1}(T_{\Delta})\geq d-1$.
\end{theorem}

In order to prove the theorem, we first prove the following useful lemma for inductively bounding the $\Delta$-treewidth of a given graph.

\begin{lemma}
\label{subgraph treewidth lemma}
For any integers $d$ and $\Delta$, if a graph $G$ contains at least $d$ disjoint connected subgraphs $G_1, G_2, \ldots, G_d$, and the $(\Delta-1)$-treewidth of each of them is at least $d-1$, then the $\Delta$-treewidth of $G$ is at least $d-1$.
\end{lemma}
\begin{proof}
Suppose $T$ is a rooted tree-decomposition of graph $G$, and $R$ is the root bag of $T$. We are going to prove that either the height of $T$ is larger than $\Delta$ or the width of $T$ is at least $d-1$. There are two cases to consider:

\begin{enumerate}
    \item $R$ contains at least one vertex from each of the subgraphs $G_1, G_2, \ldots, G_d$.
    \item $R$ does not contain any vertex from (at least) one of the subgraphs $G_1, G_2, \ldots, G_d$. 
\end{enumerate} 

In the first case, the size of $R$ is at least $d$, so the width of $T$ is at least $d-1$. In the second case, we can assume without loss of generality that $R$ does not contain any vertex from $G_1$. Let $T_1,T_2,\ldots,T_k$ be the subtrees obtained by removing $R$ from $T$. Since $R$ does not contain any vertex of $G_1$, at least one of $T_1,T_2,\ldots,T_k$ must contain some vertex from $G_1$. Suppose that subtree is $T_1$. For any vertex $v$ contained in both $T_1$ and $G_1$, since $v$ is not in the root bag $R$, it must be the case that $v$ is not contained in any other $T_2,\ldots,T_k$ as well. Similarly, every neighbor $u$ of $v$ in $G_1$ is also contained in $T_1$ as $u$ is not contained in $R$, and there must be a bag in $T$ which contains both $u$ and $v$. Proceeding this way, we get that $T_1$ contains the whole of $G_1$, and the vertices from $G_1$ appear nowhere else.

Removing vertices not in $G_1$ from each bag of $T_1$, we obtain a new tree $T_1'$. We claim that $T_1'$ is a tree-decomposition of $G_1$. Indeed, all vertices of $G_1$ are in $T_1$ and the bags in $T_1$ that contain a vertex $v$ form a connected component since $T$ was a tree decomposition of $G$. So the same holds for $T_1'$. Moreover, for every edge $(u,v)$ in $G_1$, there is a bag in $T_1$ that contain both $u$ and $v$, so the same holds for $T_1'$. 

Since the $(\Delta-1)$-treewidth of $G_1$ is at least $d-1$, either the height of $T_1'$ is larger than $\Delta-1$ or the width of $T_1'$ is at least $d-1$. Since $T_1'$ was formed by removing vertices from $T_1$, the same holds for $T_1$. Consequently, either the height of $T$ is larger than $\Delta$ or the width of $T$ is at least $d-1$. In both cases, the lemma holds.
\end{proof}

We are now ready to prove \Cref{thm:tw-lb}.

\begin{proof}[Proof of \Cref{thm:tw-lb}]
    We have $\tw_\Delta(T_{\Delta})=1$ for all $\Delta$, since $T_{\Delta}$ is a tree. For the lower bound, consider the base case $\Delta=2$. The height-$1$ tree-decomposition of the height-$2$ full $d$-ary tree $T_2$ has only one bag, and this bag contains all the vertices from $T_2$. Hence, its treewidth is $d \geq d-1$.
    
    Assume by induction that the theorem holds for all $2 \leq \Delta \leq k$ for $k \in \mathbb N$. Then, for $\Delta=k+1$, consider the height-$(k+1)$ full $d$-ary tree $T_{k+1}$. Removing the root node of $T_{k+1}$ yields $d$ pairwise disjoint height-$k$ full $d$-ary trees. By our inductive assumption, the $(k-1)$-treewidth of each of these trees is at least $d-1$. Then by \Cref{subgraph treewidth lemma}, the $k$-treewidth of $T_{k+1}$ is at least $d-1$.
\end{proof}

\section{Upper bounds in \texorpdfstring{\Cref{thm: main}}{} and \texorpdfstring{\Cref{thm: abp}}{}}
\label{sec:upper-bound}

We prove the upper bound in \Cref{thm: main}.
First, we require additional standard notation for tree-decompositions:
We consider $T$ to be rooted with a choice of root that minimizes its height.
Given a tree-decomposition of $H$ with underlying tree $T$ and bags $\{X_t\}_{t \in V(T)}$, write $\gamma(t):=\bigcup_{s\geq t}X_{s}$ for the cone at $t$,
where $s$ ranges over all descendants of $t$ in the tree $T$. 

Our second definition is more technical and specific to the dynamic programming approach we use to compute homomorphism polynomials in a bottom-up manner: It allows us to track \emph{where} in the tree-decomposition an edge contributes to a monomial of the final polynomial.
We say that an \emph{edge-representation} of $H$ in $T$ is a function $\rep:E(H)\to V(T)$ that assigns to each edge of $H$ a node in $T$ 
such that $\{u,v\}\subseteq X_{\rep(uv)}$ for all $uv\in E(H)$.
Note that each edge $uv\in E(H)$ is already entirely contained in \emph{at least one} bag by the definition of a tree-decomposition; the function $\rep$ simply chooses one such bag for each edge.

Given an edge-representation $\rep$, we define the $\rep$-height of $T$ (which will be the product-depth of the constructed circuit) as the maximum number of ``active'' nodes $t$ on a root-to-leaf path in $T$, where we call a node $t$ active
iff
\begin{itemize} 
\item there are distinct $e,e' \in E(H)$ with $\rep(e) = \rep(e') = t$, or 
\item there is at least one $e \in E(H)$ with $\rep(e) = t$ and $t$ has a child, or 
\item $t$ has at least two children. 
\end{itemize}
In our dynamic programming approach that proceeds bottom-up on a tree-decomposition, only active nodes require multiplication gates; the rep-height will thus amount to the overall product-depth of the circuit.  

\begin{lemma}
\label{Lem:UB}
Let $H$ be a graph with a tree-decomposition consisting of tree $T$ and bags $\{X_t\}_{t \in V(T)}$,
and let $\rep$ be an edge-representation of $H$ in $T$. Then there are circuits for
$\mathrm{Hom}_{H,n}$ and $\mathrm{ColIso}_{H,n}$ with product-depth equal to the $\rep$-height of $T$
and $O(|V(T)|\cdot n^{w})$ gates for $\max_{t\in V(T)}|X_t|=w$.
\end{lemma}

\begin{proof}
We describe the circuit for $\Hom_{H,n}$ and remark that the circuit for $\CI_{H,n}$ can be constructed analogously.
Considering $T$ to be rooted, and proceeding from the leaves of $T$ to the root, we inductively compute
polynomials $\mathsf{Restr}_{t,h}$ for nodes $t\in V(T)$ and functions $h:X_t\to[n]$. 
The polynomials are defined as 
\begin{align*}
\mathsf{Restr}_{t,h} & =\sum_{\substack{f:\gamma(t)\to[n]\\
f\text{ extends }h\ 
}
}\prod_{\substack{uv\in E(H)\\
\mathrm{rep}(uv) \geq t
}
}x_{f(u),f(v)}.
\end{align*}
Here, we write $s \geq t$ to denote that $s$ is a descendant of $t$ in $T$.
Note that $\mathsf{Restr}_{t,h}$ is the restriction of $\Hom_{H,n}$ to homomorphisms $f$ that extend a given homomorphism $h$ for the bag at $t$, such that only those edges feature in the monomials that are represented in the cone $\gamma(t)$.
Then $\Hom_{H,n}$ is the sum of $\mathsf{Restr}_{r,h}$ over all
$h : X_r \to [n]$ at the root $r$ of $T$.

We show how to compute the polynomials $\mathsf{Restr}_{p,h}$ for nodes $p \in V(T)$.
Let $p\in V(T)$ be a node with children
$N\subseteq V(T)$, possibly with $N = \emptyset$ if $p$ is a leaf. Assume that $\mathsf{Restr}_{t,h'}$ is known for
all $t\in N$ and functions $h' : X_t \to [n]$. Then we have
\begin{align}
\label{eq: restr-recursive}
\mathsf{Restr}_{p,h}=\left(\prod_{\substack{uv\in E(H)\\
\mathrm{rep}(uv) = p
}
}x_{h(u),h(v)}\right) \cdot \prod_{t\in N}\sum_{\substack{h':X_t\to[n]\\
\text{agreeing with }h\\
\text{on } X_t \cap X_p
}
}\mathsf{Restr}_{t,h'}.
\end{align}
From this construction of the circuit, the size bound claimed in the lemma is obvious.
Let us investigate its product-depth:
In the final circuit computing $\Hom_{H,n}$, every path from the output gate to an input gate corresponds to a path in $T$ from the root to a leaf.
Analyzing \eqref{eq: restr-recursive}, we see that every node $t$ on this path contributes $1$ to the product-depth iff $t$ is active under the edge-representation $\rep$.
Indeed, a leaf $p$ only contributes to the product-depth if two edges $e,e'\in E(H)$ are represented in its bag, i.e., $\rep(e)=\rep(e')=p$, as then there is a nontrivial product in the first product (over $uv$, shown in parentheses in \eqref{eq: restr-recursive}).
A node $p$ with one child only contributes if at least one edge is represented in its bag, as then the product between the parentheses and the remaining factor is nontrivial.
A node $p$ with at least two children always contributes to the product-depth.

To show that the circuit correctly computes $\Hom_{H,n}$, we need to show that the recursive expression for $\mathsf{Restr}_{p,h}$ in \eqref{eq: restr-recursive} is correct.
Note that every edge is represented by $\rep$ in exactly one bag and thus appears precisely once in a monomial.
Because $X_p$ is a separator in $H$, any function $f: \gamma(p) \to [n]$ gives rise to $|N|$ functions $f_t : \gamma(t)\to [n]$ for $t\in N$ that all agree on their values for $X_p$ (that is, on their values on $X_p \cap X_t$) and can otherwise be chosen independently. Conversely, any ensemble of such consistent functions can be merged to a function $h: \gamma(p) \to [n]$.
The product over all $t\in N$ as in \eqref{eq: restr-recursive} thus yields $\mathsf{Restr}_{p,h}$.
\end{proof}

Finally, to prove the upper bound in \Cref{thm: main},
let $H^\dagger$ be the graph obtained from $H$ by removing all degree-$1$ vertices.
Given a tree-decomposition for $H^\dagger$ with underlying tree $T$ of height $\Delta$ and width $w$ witnessing that $\ptw_\Delta(H) = w$,
we obtain a tree-decomposition with some tree $T'$ for $H$ and an edge-representation $\mathrm{rep}$ of $H$ in $T'$ of rep-height $\Delta$ as follows:
For each vertex $v \in V(H)$ of degree $1$, with neighbor $u \in V(H)$, choose some node $t \in T$ with $u\in X_t$ and add a node $t'$ as a neighbor of $t$ to $T$ with bag $X_{t'} = \{v,u\}$.
Choose an arbitrary representation $\mathrm{rep}$ of $H$ in the resulting tree-decomposition with tree $T'$ and observe that its $\mathrm{rep}$-height is at most the height $\Delta$ of $T$, even though the height of $T'$ may be $\Delta+1$: The bags added for degree-$1$ vertices and their incident edges do not contribute towards the $\mathrm{rep}$-height, as they are leaf nodes and represent single edges.
The upper bound thus follows from \Cref{Lem:UB}.

\begin{remark}
\label{rem:ub-abp}
The construction from Lemma~\ref{Lem:UB} also yields an ABP of length $|V(T)|$ and size $O(|V(T)|\cdot n^w)$ when given a path-decomposition $T$ of $H$ with maximum bag size $w$. To see this, note that the product over $t \in N$ in \eqref{eq: restr-recursive} involves only a single factor when $T$ is a path-decomposition, so \eqref{eq: restr-recursive} overall amounts to a skew-multiplication of a single monomial with the recursively computed polynomial.
\end{remark}

\section{Lower bounds in \texorpdfstring{\Cref{thm: main}}{} and \texorpdfstring{\Cref{thm: abp}}{}}
\label{sec:lower-bound}

We adapt the lower bound proofs of \cite{KPR2023} to prove the lower bounds in our theorems. Recall that proving the lower bound for $\CI_{H,n}$ is enough, since we can use a circuit computing $\Hom_{H,n}$ to obtain a circuit computing $\CI_{H,n}$ without changing the depth of the circuit using~\cite[Lemma 8]{KPR2023}.  We summarize the results here for completeness.

\begin{lemma}
\label{lem: col-hom-reduction}
    Let $k,\Delta$ be positive integers and $H$ be a fixed pattern graph on $k$ vertices.
    \begin{itemize}
    \item If there is a monotone circuit of product-depth $\Delta$ and size $s$ for $\CI_{H,n}$, then there is such a circuit of size $O(s)$ for $\Hom_{H,n}$.
    \item If there is a monotone circuit of product-depth $\Delta$ and size $s$ for $\Hom_{H,n'}$, then there is such a circuit of size $O(s^{|E(H)|})$ for $\CI_{H,n}$, where $n'=kn$.
    \end{itemize}
    The results also hold if circuits are replaced by ABPs, and product-depth is replaced by the length of the ABP, provided the length is at least the degree of the polynomials.
\end{lemma}
\begin{proof}
    Given a monotone circuit of product-depth $\Delta$ that computes $\CI_{H,n}$, we replace each variable $x^{(uv)}_{f(u),f(v)}$ with $x_{f(u),f(v)}$ if $f(u) \neq f(v)$ and $0$ otherwise. The circuit now computes $\Hom_{H,n}$.

    For the other direction, let $C$ be the monotone circuit of product-depth $\Delta$ computing the $\Hom_{H}$ polynomial over the vertex set $[k] \times [n]$. Note that a homomorphism $\phi$ from $H$ to the complete graph on $[k] \times [n]$ maps a vertex $u \in [k]$, to $(v,p)$ where $v \in [k]$ and $p \in [n]$. We introduce auxiliary variables $y_{uv}$ for each edge $uv \in E(H)$. For $u,v \in [k]$ and $p,q \in [n]$, we replace the variable $x_{(u,p),(v,q)}$ with $x_{p,q}^{(uv)} y_{uv}$ if $uv \in E(H)$ and $0$ otherwise. 
    
    Let $C'$ be the new circuit obtained after the replacement, and consider the partial derivative ${D:= \frac{\partial^{|E(H)|}}{\partial y_{e_1} \cdots \partial y_{e_{|E(H)|}}} C'}$, with respect to all the edge variables of $H$. Note that every monomial in $D$ contains at least one variable corresponding to each edge of $H$. Further, set $y_{uv}=0$ in $D$ for all $uv \in E(H)$. This ensures that every monomial in $D|_{y_{uv}=0}$ contains exactly one variable corresponding to every edge of $H$, i.e., it counts only the color-preserving homomorphisms. The coefficient of each monomial is $|aut(H)|$, the number of automorphisms of $H$, and dividing by this number gives us $\CI_{H,n}$.
    
    We can compute $D$ using partial derivatives' sum and product rules applied to every gate in a bottom-up fashion. For a gate $g$, we maintain both $g$ and $\partial_{y_e} g$.  The partial derivative of a sum gate, $\partial_{y_e} \sum_i g_i = \sum_i \partial_{y_e}g_i$ is straightforward and does not increase the depth. For a product gate, the derivative $\partial_{y_e} \prod_i g_i = \sum_i \left(\partial_{y_e}g_i \prod_{j \neq i} g_j \right)$ increases the depth by one, but this can be absorbed in the sum layer above. Note that the product-depth does not change in both cases. A partial derivative with respect to a single variable increases the circuit size by a factor of $s$. Hence, the final circuit for $D$ is of size $O(s^{|E(H)|})$, and has product-depth $\Delta$, the same as $C$.

    We also note that both the constructions preserve monotonicity. Moreover, if the original circuit $C$ was \emph{skew} (i.e. an ABP), then so is the final circuit $D$. From \Cref{rem:abp-skew}, we obtain the same results for ABPs as well.
\end{proof}

\subsection{Tree decompositions from parse trees}

Consider a pattern graph $H$ on vertex set $V(H) := [k]$. 
An alternative and more intuitive way to think about the {$n$-th} colored subgraph isomorphism polynomial $\CI_{H,n}$ is to consider the blown-up graph $G$, where each vertex $u \in [k]$ of $H$ is replaced by a `cloud' of $n$ vertices $C_u := \{(u,1),\ldots,(u,n)\}$. 
Every edge $uv \in E(H)$ is replaced by a complete bipartite graph between $C_u$ and $C_v$ with an appropriate label for each of the $n^2$ edges; that is, an edge between $(u,i)$ and $(v,j)$ is labeled $x^{(uv)}_{i,j}$ where $u,v \in [k]$ and $i,j \in [n]$. The polynomial $\CI_{H,n}$ is now obtained by choosing a copy of $H$ in $G$ by picking a vertex from every cloud using a function $f:V(H) \to [n]$, and adding the monomial 
$$m = \prod_{uv \in E(H)} x^{(uv)}_{f(u),f(v)}.$$ 

We say that the monomial $m$ above is supported on a set $S \subseteq [k] \times [n]$ if every element of $S$ looks like $(u,f(u))$ for $u \in [k]$. The polynomial $\CI_{H,n}$ is the sum over all such monomials $m$
$$\CI_{H,n} = \sum_{f:V(H) \to [n]} \prod_{uv \in E(H)} x^{(uv)}_{f(u),f(v)}.$$

\begin{claim}
\label{clm:td-extract}
    Let $\Delta$ be a natural number and $\T$ be a monotone parse tree of product-depth $\Delta$ computing a monomial $m$ of $\CI_{H,n}$. Let $H^{\dagger}$ be the pruned graph obtained by removing all degree-$1$ vertices from $H$. We can extract from $\T$ a tree decomposition of $H^{\dagger}$ with underlying tree $T^{\dagger}$ of height $\Delta$.
\end{claim}
\begin{proof}
    Suppose that the monomial $m$ is supported on vertices $(u,f(u))$ where $u \in [k]$ and $f:[k] \to [n]$ is a function. The parse tree $\T$ has height $\Delta+1$. Note that since $\CI_{H,n}$ has $0/1$ coefficients, we can assume that a multiplication gate has only non-constant terms as its children. We build the tree decomposition bottom-up. We `mark' certain vertices in the bags created during this procedure. All such marks are dropped at the end (see \Cref{fig:tree-decomp}). 

    \begin{enumerate}
        \item For an input gate $x_{f(u),f(v)}^{(uv)}$, we add the bag $\{u,v\}$ as a leaf in the tree decomposition. We \emph{mark} all the vertices of degree $1$. The rest are \emph{unmarked}.
        
        \item Let $g$ be a multiplication gate. Suppose $X_1, \ldots, X_m$ are the bags corresponding to the children of $g$ (that we have already constructed) and let $U_i \subseteq X_i$ be the \emph{unmarked} elements of $X_i$. We then add the bag $X_g := \bigcup_{i \in [m]} U_i$ as the root of $X_1,\ldots,X_m$. If there are vertices $(u,f(u))$ such that the monomial computed at $g$ includes all the edges incident on $(u,f(u))$ in the copy of $H$ that $f$ picked, we \emph{mark} all such vertices $u$ in the bag $X_g$.

        \item Finally, after applying the procedure in the previous step to all the gates, we drop the bags (and edges) corresponding to input gates.
    \end{enumerate}

    We claim that the tree decomposition we just constructed with underlying tree $T^{\dagger}$ and bags $\{X_u\}_{u \in V(T^{\dagger})}$ is a tree decomposition of $H^{\dagger}$. Note that \emph{all} the edges of $H$ were covered at the leaf bags (that we finally dropped), as they must be present in the monomial. Since only the degree-$1$ vertices in a leaf bag were \emph{marked}, the parent bags of the leaves (which we include in our tree decomposition) will exactly have the vertices of $H^{\dagger}$, and thus cover all its edges.
    
    We mark (forget) a vertex only after multiplying all its incident edges. Hence, the sub-graph induced by a vertex $u$ (in $H^{\dagger}$) is \emph{connected} in $T^{\dagger}$ and is, in fact, a subtree. As every multiplication gate of the parse tree has exactly one associated bag, the procedure does indeed result in a tree decomposition of $H^{\dagger}$ of height $\Delta$.
\end{proof}

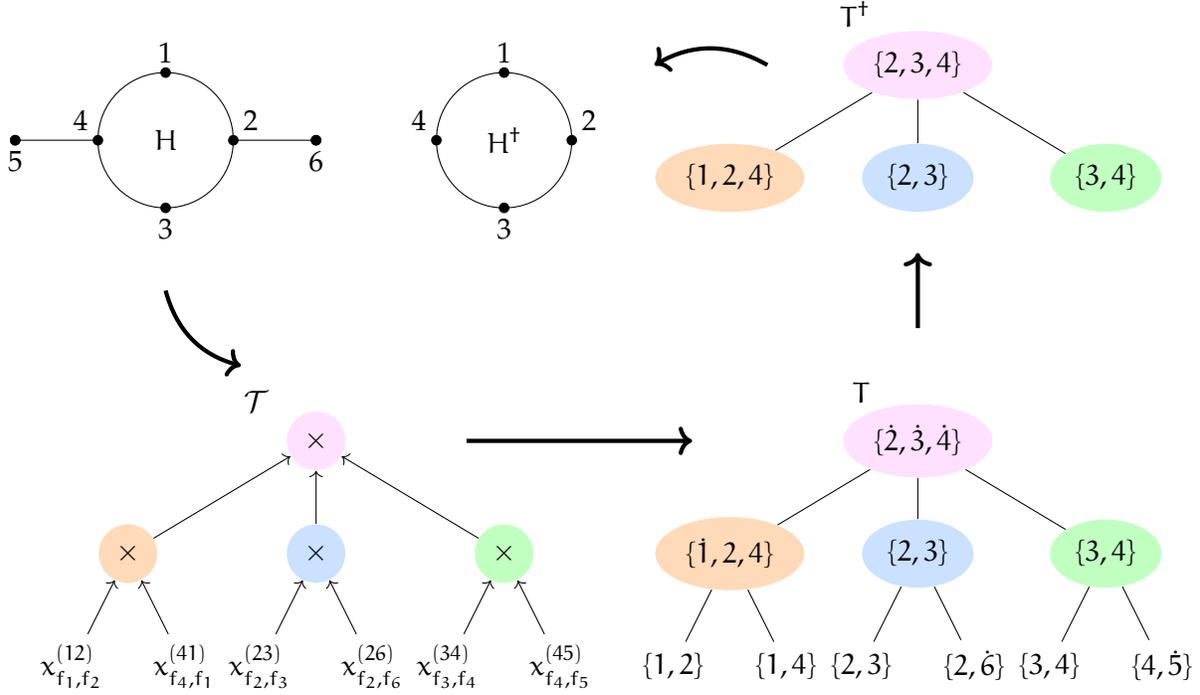
\begin{figure}
    \centering
\begin{tikzpicture}
    \begin{scope}[xshift=-2cm]
      \def\r{0.9}
    
      \coordinate (v1) at (90:\r);    
      \coordinate (v2) at (0:\r);     
      \coordinate (v3) at (270:\r);   
      \coordinate (v4) at (180:\r);   
    
      \coordinate (v5) at (2.0,0);    
      \coordinate (v6) at (-2.0,0);   
    
      \draw (v1) to[bend left=45] (v2);
      \draw (v2) to[bend left=45] (v3);
      \draw (v3) to[bend left=45] (v4);
      \draw (v4) to[bend left=45] (v1);
    
      \draw (v2) -- (v5);
      \draw (v4) -- (v6);
    
      \foreach \coord in {v1,v2,v3,v4,v5,v6} {
        \fill (\coord) circle (2pt);
      }
    
      \node[above]       at (v1) {1};
      \node[above right] at (v2) {2};  
      \node[below]       at (v3) {3};
      \node[above left]  at (v4) {4};  
      \node[below]   at (v6) {5};  
      \node[below]    at (v5) {6}; 
      \node at (0,0) {$H$};
    
    \end{scope}

    \draw[->, bend right, ultra thick] (-2,-2) to (-1,-3);

    \begin{scope}[xshift=-3cm,yshift=-3cm,scale=0.50]
    
    \node[circle,fill=Thistle1] (g1) at (6.0, -2) {$\times$};
    \node[above left] at (5,-1.5) {$\mathcal{T}$};
    \node[circle,fill=PeachPuff1] (g2) at (1.0, -5) {$\times$} edge[->] (g1);
    \node[circle,fill=LightSteelBlue1] (g3) at (6.0, -5) {$\times$} edge[->] (g1);
    \node[circle,fill=DarkSeaGreen1] (g4) at (11.0, -5) {$\times$} edge[->] (g1);
    
    \node[rectangle] (x12) at (-0.5,-8) {$x^{(12)}_{f_1,f_2}$} edge[->] (g2);
    \node[rectangle] (x41) at (2.5,-8) {$x^{(41)}_{f_4,f_1}$} edge[->] (g2);

    \node[rectangle] (x23) at (4.5,-8) {$x^{(23)}_{f_2,f_3}$} edge[->] (g3);
    \node[rectangle] (x26) at (7.5,-8) {$x^{(26)}_{f_2,f_6}$} edge[->] (g3);

    \node[rectangle] (x34) at (9.5,-8) {$x^{(34)}_{f_3,f_4}$} edge[->] (g4);
    \node[rectangle] (x45) at (12.5,-8) {$x^{(45)}_{f_4,f_5}$} edge[->] (g4);
    \end{scope}

    \draw[->,ultra thick] (2,-4) to (5,-4);

    \begin{scope}[xshift=5cm,yshift=-3cm,scale=0.5]
        \node[ellipse,fill=Thistle1] (g1) at (6.0, -2) {$\{\dot 2,\dot 3,\dot 4\}$};
        \node[above left] at (5,-1.25) {$T$};
    \node[ellipse,fill=PeachPuff1] (g2) at (1.0, -5) {$\{\dot 1,2,4\}$} edge[-] (g1);
    \node[ellipse,fill=LightSteelBlue1] (g3) at (6.0, -5) {$\{2,3\}$} edge[-] (g1);
    \node[ellipse,fill=DarkSeaGreen1] (g4) at (11.0, -5) {$\{3,4\}$} edge[-] (g1);

    \node[rectangle] (x12) at (-0.5,-8) {$\{1,2\}$} edge[-] (g2);
    \node[rectangle] (x41) at (2.5,-8) {$\{1,4\}$} edge[-] (g2);

    \node[rectangle] (x23) at (4.5,-8) {$\{2,3\}$} edge[-] (g3);
    \node[rectangle] (x26) at (7.5,-8) {$\{2,\dot 6\}$} edge[-] (g3);

    \node[rectangle] (x34) at (9.5,-8) {$\{3,4\}$} edge[-] (g4);
    \node[rectangle] (x45) at (12.5,-8) {$\{4,\dot 5\}$} edge[-] (g4);
    \end{scope}

    \draw[->, ultra thick] (8,-2.5) to (8,-1.5);

    \begin{scope}[xshift=5cm,yshift=2cm,scale=0.5]
        \node[ellipse,fill=Thistle1] (g1) at (6.0, -2) {$\{2,3,4\}$};
        \node[above left] at (5,-1.25) {$T^{\dagger}$};
    \node[ellipse,fill=PeachPuff1] (g2) at (1.0, -5) {$\{1,2,4\}$} edge[-] (g1);
    \node[ellipse,fill=LightSteelBlue1] (g3) at (6.0, -5) {$\{2,3\}$} edge[-] (g1);
    \node[ellipse,fill=DarkSeaGreen1] (g4) at (11.0, -5) {$\{3,4\}$} edge[-] (g1);
    \end{scope}

    \draw[->, ultra thick, bend right] (6,1) to (4.5,1);

    \begin{scope}[xshift=2.5cm]
        \def\r{0.9}
    
      \coordinate (v1) at (90:\r);    
      \coordinate (v2) at (0:\r);     
      \coordinate (v3) at (270:\r);   
      \coordinate (v4) at (180:\r);   
    
      \draw (v1) to[bend left=45] (v2);
      \draw (v2) to[bend left=45] (v3);
      \draw (v3) to[bend left=45] (v4);
      \draw (v4) to[bend left=45] (v1);
    
      \foreach \coord in {v1,v2,v3,v4,v5,v6} {
        \fill (\coord) circle (2pt);
      }
    
      \node[above]       at (v1) {1};
      \node[above right] at (v2) {2};  
      \node[below]       at (v3) {3};
      \node[above left]  at (v4) {4};  
      \node at (0,0) {$H^{\dagger}$};
    \end{scope} 
\end{tikzpicture}
\caption{Extracting a tree decomposition of height $2$ for $H^{\dagger}$ from a parse-tree of product-depth $2$ for a monomial of $\CI_{H,n}$. We have for all $i\in[6]$, $f_i := f(i) \in [n]$.}
\label{fig:tree-decomp}
\end{figure}

\subsection{Lower bounds for \texorpdfstring{$\CI_{H,n}$}{ColSub}}

\begin{theorem}
\label{thm:mon-lb}
    Let $\Delta$ be a natural number and $H$ be a pattern graph. Any monotone circuit of \emph{product-depth} $\Delta$ computing the polynomial $\CI_{H,n}$ has size $\Omega(n^{\ptw_{\Delta}(H)+1})$.
\end{theorem}

\begin{proof}
    Let $C$ be the monotone circuit computing $\CI_{H,n}$, and let the pruned $\Delta$-treewidth of $H$, $\ptw_{\Delta}(H) = t$. Consider a monomial $m$ of $\CI_{H,n}$ supported on vertices $(u,f(u))$ for $u \in [k]$ and $f:[k] \to [n]$. Let $\T$ be a parse tree of $C$ associated with $m$. Now, \Cref{clm:td-extract} gives a tree decomposition of $H^{\dagger}$ with tree $T^{\dagger}$ and bags $\{X_u\}_{u \in V(T^{\dagger})}$. Consequently, there is a bag $X$ of size at least $t+1$ in the tree decomposition. Without loss of generality, we assume that $|X|=t+1$. If it is greater, we will only obtain a better lower bound. We also assume that the vertices in the bag are $1,\ldots,t+1$ (relabeling the vertices of $H$ does not change the complexity of $\CI_{H,n}$). Let the corresponding gate in $\T$ associated with $X$ be $g$. 
    
    We show that only a `few' monomials can contain $g$ in their parse tree. More precisely, we claim that any monomial $m'$ (other than $m$) that contains $g$ in its parse tree is supported on vertices $\{(u,f(u))\}_{u \in [t+1]}$. Suppose not. Let $m'$ have a parse tree $\T'$ with gate $g$ in it but vertex $(u,f'(u))$ for some $u \in [t+1]$, with $f(u) \neq f'(u)$. Recall that we obtained the tree decomposition using the parse tree $\T$ of $m$. For a gate $g$ in a parse tree, we denote by $\T_g$ the subtree rooted at $g$. Note that if two parse trees contain a multiplication gate $g$, all the children of $g$ are the same in both the parse trees. We now analyze two cases:
    \begin{enumerate}
        \item The vertex $u$ is marked at the bag associated with $g$: There are at least two children $g_1,g_2$ of $g$ in $\T$ that compute monomials with $(u,f(u))$ in them. This holds because there are no degree-$1$ vertices in the bags. If $g_1$ in $\T'$ contains the vertex $(u,f'(u))$, we replace $\T'_{g_2}$ with $\T_{g_2}$. Similarly, in the other case, when $g_2$ contains $(u,f'(u))$. If both $g_1,g_2$ do not contain $(u,f'(u))$ in $\T'$, we arbitrarily replace $\T'_{g_1}$ (say) with $\T_{g_1}$.
        \item The vertex $u$ is not marked at the bag associated with $g$: The vertex $(u,f(u))$ appears in $\T_g$ \emph{as well as} outside $\T_g$. In $\T'$, if $(u,f'(u))$ appears in $\T'_{g}$, we replace $\T_g$ with $\T'_g$ in $\T$. Otherwise, we replace $\T'_g$ with $\T_g$ in $\T'$.
    \end{enumerate}
    
    In all cases, we obtain a valid parse tree $\T''$ of $C$ that produces a monomial supported on $(u,f(u))$ \emph{and} $(u,f'(u))$. This leads to a contradiction, since the monomial produced by $\T''$ is spurious and cannot be cancelled because the circuit is monotone. Every monomial (parse tree) $m$ has a gate $g$ whose corresponding bag has at least $t+1$ vertices. And any other monomial $m'$ (parse tree) that contains this gate $g$ must share at least $t+1$ vertices in its support with $m$. Thus, the maximum number of monomials containing this gate $g$ equals the number of colored isomorphisms that fix $t+1$ vertices, which is $n^{k-t-1}$. Recall that there are $n^k$ monomials in $\CI_{H,n}$, and so we need at least $n^{t+1}$ gates in the circuit.
\end{proof}

The lower bound proof for algebraic branching programs is very similar.

\begin{theorem}
    Let $\Delta$ be a positive integer and $H$ be a pattern graph such that $\Delta \geq |E(H)|$. Any monotone ABP of \emph{length} $\Delta$ computing the polynomial $\CI_{H,n}$ has size $\Omega(n^{\ppw_{\Delta}(H)+1})$.
\end{theorem}
\begin{proof}
    As mentioned earlier in \Cref{rem:abp-skew}, the size-$s$ monotone ABP of length $\Delta$ computing $\CI_{H,n}$ has an equivalent monotone skew-circuit $C$ of size $O(s)$ and product-depth $\Delta$.
    Consider a monomial $m$ of $\CI_{H,n}$ supported on vertices $(u,f(u))$ for $u \in [k]$ and $f:[k] \to [n]$. Let $\T$ be a parse tree of $C$ associated with $m$. 
    
    We observe that the procedure described in the proof of \Cref{clm:td-extract} extracts a length-$\Delta$ \emph{path decomposition} of the pruned graph $H^{\dagger}$ instead: as the circuit is skew, all the multiplication gates in $\T$ have at most one non-leaf child. Since we finally dropped the bags corresponding to input gates in our procedure, the tree decomposition we obtain is in fact a path decomposition! 
    
    Taking the pruned $\Delta$-pathwidth of $H$ to be $t$, the same proof implies that the number of monomials containing a particular gate $g$ is $n^{k-t-1}$, thus implying a size lower bound of $n^{t+1}$.
\end{proof}

\section{Depth Hierarchy}

Combining our previous results allows us to prove a depth-hierarchy theorem for bounded-depth monotone algebraic circuits.\footnote{A similar hierarchy can also be shown for monotone ABPs.} 

\begin{theorem}~\label{thm:mlbw}
    For all integers $n$ and $\Delta$, there exists a pattern graph $H_{\Delta}$ such that $\CI_{H_{\Delta},n}$ can be computed by a monotone product-depth $(\Delta+1)$ circuit of size $ O(n|H_{\Delta}|)$ but any product-depth $\Delta$ monotone circuit computing the polynomial needs size $n^{\Omega(|H|^{1/\Delta})}$.
\end{theorem}
\begin{proof}
    For an integer $d$, let $H_{\Delta} := T_{\Delta+2}$ be the full $d$-ary tree of height $\Delta+2$. Note that $d=\Theta(|H_{\Delta}|^{1/\Delta})$. The pruned $(\Delta+1)$-treewidth of $H_{\Delta}$ is equal to the $(\Delta+1)$-treewidth of the full $d$-ary tree of height $(\Delta+1)$. That is, $\ptw_{\Delta+1}(H_{\Delta})=\tw_{\Delta+1}(T_{\Delta+1})=1$. So by \Cref{Lem:UB}, there exists a monotone circuit of product-depth $\Delta+1$ and size $O(n|H_{\Delta}|)$, which computes $\CI_{H_{\Delta},n}$. 
    
    On the other hand, we have $\ptw_\Delta(H_{\Delta})=\tw_\Delta(T_{\Delta+1}) \geq d-1$ by \Cref{thm:tw-lb}. Hence, by \Cref{thm:mon-lb}, every monotone circuit of product-depth $\Delta$ computing $\CI_{H_{\Delta},n}$ necessarily has size at least $\Omega(n^{\ptw_\Delta(H_{\Delta})+1}) = n^{\Omega(|H_{\Delta}|^{1/\Delta})}$.
\end{proof}

If we consider the case when the pattern graph is of size $\Theta(n)$ in \Cref{thm:mlbw}, we obtain the depth hierarchy result in \Cref{thm:dh}. As an aside, we note that an analogous depth hierarchy cannot be obtained for the polynomials $\Hom_{H,n}$ using our methods, as the blow up in the size given by \Cref{lem: col-hom-reduction} is exponential, when $|H|=\Theta(n)$ is not a constant.

\section{Acknowledgements}

We want to thank Vishwas Bhargava, Cornelius Brand, Deepanshu Kush, and Anurag Pandey for insightful discussions during this project.  We also thank Marcin Pilipczuk for pointing out the graph parameter \emph{vertex integrity}.

Part of this work was carried out during the Copenhagen Summer of Counting \& Algebraic Complexity, funded by research grants from VILLUM FONDEN (Young Investigator Grant 53093) and the European Union (ERC, CountHom, 101077083). Views and opinions expressed are those of the authors only and do not necessarily reflect those of the European Union or the European Research Council Executive Agency. Neither the European Union nor the granting authority can be held responsible for them.


\renewcommand\bibname{References}
\printbibliography

\end{document}